\documentclass[aps,pra,twocolumn,superscriptaddress,nofootinbib]{revtex4}

\usepackage{graphicx}
\usepackage[normal]{subfigure}
\usepackage{caption}
\usepackage{latexsym}
\usepackage{amsmath}
\usepackage{amssymb}
\usepackage{amsfonts}
\usepackage{color}
\usepackage{pgf, tikz}
\usepackage{dsfont}
\usepackage{caption}
\usepackage{amsthm}

\usepackage[utf8]{inputenc}
\usepackage[english]{babel}

\begin{document}

\newcommand{\ket}[1]{\vert #1 \rangle}
\newcommand{\bra} [1] {\langle #1 \vert}
\newcommand{\braket}[2]{\langle #1 | #2 \rangle}
\newcommand{\proj}[1]{\ket{#1}\bra{#1}}
\newcommand{\mean}[1]{\langle #1 \rangle}
\newcommand{\opnorm}[1]{|\!|\!|#1|\!|\!|_2}
\newtheorem{theo}{Theorem}
\newtheorem{lem}{Lemma}
\newtheorem{defin}{Definition}
\newtheorem{corollary}{Corollary}
 \newtheorem{conj}{Conjecture}
 \newtheorem*{prop}{Properties}
 
\newtheorem{theobis}{Theorem}

\title{Multidimensional  entropic uncertainty relation based on a commutator matrix \\
in position and momentum spaces}

\author{Anaelle Hertz}
\email{ahertz@ulb.ac.be}
\affiliation{Centre for Quantum Information and Communication, \'Ecole polytechnique de Bruxelles, Universit\'e libre de Bruxelles, 1050 Brussels, Belgium}

\author{Luc Vanbever}
\affiliation{Centre for Quantum Information and Communication, \'Ecole polytechnique de Bruxelles, Universit\'e libre de Bruxelles, 1050 Brussels, Belgium}

\author{Nicolas J. Cerf}
\affiliation{Centre for Quantum Information and Communication, \'Ecole polytechnique de Bruxelles, Universit\'e libre de Bruxelles, 1050 Brussels, Belgium}

\begin{abstract}
The uncertainty relation for continuous variables due to Byalinicki-Birula and Mycielski expresses the complementarity between two $n$-tuples of canonically conjugate variables $(x_1,x_2,\cdots x_n)$ and $(p_1,p_2,\cdots p_n)$ in terms of Shannon differential entropy. Here, we consider the generalization to variables that are not canonically conjugate and derive an entropic uncertainty relation expressing the balance between any two $n$-variable Gaussian projective measurements. The bound on entropies is expressed in terms of the determinant of a matrix of commutators between the measured variables. This  uncertainty relation also captures the complementarity between any two incompatible linear canonical transforms, the bound being written in terms of the corresponding symplectic matrices in phase space. Finally, we extend this uncertainty relation to Rényi entropies and also prove a covariance-based uncertainty relation which generalizes Robertson relation.



\end{abstract}

\maketitle


\section{Introduction}

At the heart of quantum mechanics, uncertainty relations reflect the impossibility to define -- exactly and simultaneously -- the value of two observables that do not commute, such as the position $\hat x$ and momentum $\hat p$ of a particle. Uncertainty relations, expressed in terms of variances of observables, were first introduced by Heisenberg \cite{Heisenberg} and Kennard \cite{Kennard}, and then generalized by Schr\"odinger \cite{Schrodinger} and Robertson \cite{Robertson}. Later on, it was shown by Hirschman \cite{Hirschman} that uncertainty relations may also be formulated in terms of Shannon entropies instead of variances, leading to the first entropic uncertainty relation for canonically conjugate variables $\hat x$ and $\hat p$ proven by Bialynicki-Birula and Mycielski \cite{Birula} and Beckner \cite{Beckner}. Entropic uncertainty relations have also been developed for discrete observables in finite-dimensional spaces, see  \cite{Coles} for a review, but here we focus on continuous-spectrum observables in an infinite-dimensional space. Specifically, we use the notations of quantum optics and view variables $\hat x$ and $\hat p$ as canonically conjugate quadrature components of a bosonic mode.
Then, the $n$-modal version of the entropic uncertainty relation for the $n$-tuples $\vec x = (x_1,x_2,\cdots x_n)$ and $\vec p=(p_1,p_2,\cdots p_n)$ is expressed as\footnote{We set $\hbar=1$ throughout this paper.}\cite{Birula} 
\begin{equation}
h(\vec x) + h(\vec p) \geq n \ln(\pi e )
\label{birula}
\end{equation}
 where $h(\vec x)$ and $h(\vec p)$ are the Shannon differential entropies of  $\vec x$ and  $\vec p$, namely
 \begin{eqnarray}
h(\vec x) \equiv h(|\psi(\vec x)|^2 )&=& - \int d\vec x\; |\psi(\vec x)|^2 \ln |\psi(\vec x)|^2,\nonumber\\
h(\vec p)  \equiv h(|\phi(\vec p)|^2 )&=& -\int d\vec p\; |\phi(\vec p)|^2 \ln |\phi(\vec p)|^2,
 \end{eqnarray}
with  $|\psi(\vec x)|^2=|\langle x_1,x_2,\cdots x_n | \psi \rangle|^2$ and $|\phi(\vec p)|^2=|\langle p_1,p_2,\cdots p_n | \psi \rangle|^2$ being the probability distributions of $\vec x$ and  $\vec p$ in the pure state $|\psi\rangle$. Of course, $\phi(\vec p)$ is the Fourier transform of $\psi(\vec x)$, which is at the origin of the complementarity between $\vec x$ and  $\vec p$ expressed by Eq.~(\ref{birula}). 
Lately, this entropic uncertainty relation has been extended by taking $x$-$p$ correlations into account \cite{hertz}, the significant advantage being that the resulting uncertainty relation is saturated by any pure $n$-modal Gaussian state.

In 2011, Huang \cite{huang} generalized the entropic uncertainty relation to a pair of observables that are not canonically conjugate. More precisely, defining the observables
\begin{equation}
\hat{A}=\sum_{i=1}^n(a_i \, \hat{x}_i+a_i' \, \hat{p}_i),\quad
\hat{B}=\sum_{i=1}^n(b_i \, \hat{x}_i+b_i' \, \hat{p}_i),
\label{ABdef}
\end{equation} 
he showed that
\begin{equation}
h(\hat{A})+h(\hat{B})\geq\ln(\pi e |[\hat{A},\hat{B}]|)
\label{huang}
\end{equation}
where $[\hat{A},\hat{B}]$ (which is a scalar) is the commutator between both observables. Obviously, if $\hat{A}=\hat{x}$ and $\hat{B}=\hat{p}$, this inequality reduces to Eq.~(\ref{birula}). In addition, a similar result had earlier been obtained by Guanlei et al. \cite{Guanlei} in the special case where $n=1$, namely 
\begin{equation}
h(\hat{x}_\theta)+h(\hat{x}_\phi)\geq \ln(\pi e | \sin(\theta-\phi)|) 
\label{eq-Guanlei}
\end{equation}
where $\hat{x}_\theta=\hat{x} \cos \theta + \hat{p} \sin \theta$ and $\hat{x}_\phi= \hat{x} \cos \phi + \hat{p} \sin \phi $ are two rotated quadratures.

In this paper, we introduce a generalization of the uncertainty relation of Byalinicki-Birula and Mycielski, which is stated in the form of our Theorem \ref{theoFRFTcom}. It addresses the situation where $n$ arbitrary quadratures are  jointly measured on $n$ modes, expressing the balance between two such joint measurements (see Fig.~\ref{transfoAB}). In other words, we state an entropic uncertainty relation between two arbitrary $n$-modal Gaussian projective measurements (or, equivalently, two $n$-mode Gaussian unitaries $U_A$ and $U_B$). The lower bound of our uncertainty relation, Eq.~(\ref{EURFRFTcom}), depends on the determinant of a $n\times n$ matrix formed with the commutators between the $n$ measured quadratures in both cases. 
 In contrast, Eq.~(\ref{birula}) is restricted to the case of measuring either all $x$ quadratures or all $p$ quadratures on $n$ modes, while Refs.~\cite{huang,Guanlei} treat the balance between two single-mode measurements only.
 
Interestingly, the probability distribution of the measured quadratures is given by the squared modulus of the linear canonical transform (LCT) associated with $U_A$ or $U_B$, so that our entropic uncertainty relation also captures the complementarity between two incompatible $n$-dimensional LCTs as expressed by our Lemma~\ref{theoFRFT}. It simply reduces to Eq. (\ref{birula}) when the two LCTs are connected by a $n$-dimensional Fourier transform, mapping $\vec x = (x_1,x_2,\cdots x_n)^T$ onto $\vec p=(p_1,p_2,\cdots p_n)^T$.
 
In Section II, we define general $n$-dimensional LCT's and give some useful properties. Section III presents our results on uncertainty relations for $n$ modes. First, we derive a generalized  entropic uncertainty relation based on differential Shannon entropies (our Theorem \ref{theoFRFTcom}, with its extension to a larger-dimensional space), then we extend it to Rényi entropies (our Theorem \ref{theo-Renyi}), and finally we exhibit a covariance-based uncertainty relation (our Theorem \ref{theo-covariance}). In Section IV, we conclude and suggest a conjecture for a generalized entropic uncertainty relation in the case where the commutators differ from scalars.

	\begin{figure}[t]
		\includegraphics[trim = 0.5cm 19cm 0.3cm 2.5cm, clip, width=0.95\columnwidth]{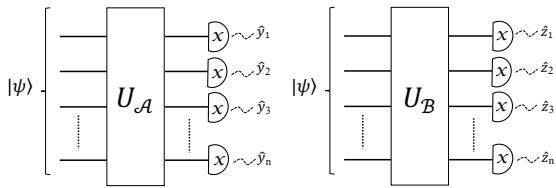}
		\caption{\label{transfoAB} Schematic of two $n$-modal Gaussian projective measurements applied onto state $|\psi\rangle$, resulting in the $n$ quadratures $\hat y_i$'s or $\hat z_i$'s. These measurements can be implemented by applying a Gaussian unitary  ($U_A$ or $U_B$) onto $|\psi\rangle$ and measuring the $\hat x$-quadratures of the $n$ modes. Our uncertainty relation, Eq.~(\ref{EURFRFTcom}), expresses the complementarity between the $\hat y_i$'s and $\hat z_i$'s, or equivalently between the linear canonical transforms associated with $U_A$ and $U_B$.}
	\end{figure}

\section{Linear canonical transforms}

Before deriving our uncertainty relations, we need to properly define fractional Fourier transforms (FRFTs) along with their generalization to LCTs. Some early papers on FRFTs appeared in the 1920's, but this topic became investigated in depth only more recently in the fields of signal processing and quantum optics (see, e.g., \cite{Pei,Bultheel,Morsche,Ozaktasbook,Ozaktas} for more details). In one dimension, the FRFT of a wave function $f(x)$ can be understood as the new wave function obtained when the Wigner function corresponding to $f(x)$ undergoes a rotation of angle $\alpha$ in phase space. If $\alpha=\pi/2$, then the FRFT simply coincides with the usual Fourier transform, connecting the time and frequency domains in the field of signal processing or the canonically conjugate $x$- and $p$-quadratures in quantum optics. Mathematically, the one-dimensional FRFT of function $f(x)$ is defined as
\begin{equation}
\mathcal{F}_\alpha(y)=\sqrt{\frac{1-i\cot\alpha}{2\pi}} \; e^{\frac{i}{2}y^2\cot\alpha}\int e^{\frac{-iyx}{\sin\alpha}}e^{\frac{i}{2}x^2\cot\alpha} f(x) \, dx
\end{equation}
The one-dimensional FRFT can be generalized to one-dimensional LCTs by including all affine linear transformations in phase space $(x,p)$, going beyond rotations. Accordingly, the LCT of wave function $f(x)$ is the new wave function obtained when the corresponding Wigner function undergoes a symplectic transformation $\mathcal{S}$. The one-dimensional LCT of $f(x)$ is defined as
\begin{equation}
\mathcal{F}_\mathcal{S}(y)=\sqrt{\frac{1}{2\pi i b}} \; e^{\frac{id}{2b}y^2}\int e^{\frac{-iyx}{b}}e^{\frac{ia}{2b}x^2} f(x)\, dx
\end{equation}
where $\mathcal{S}=\begin{pmatrix}a&b\\c&d\end{pmatrix}$ is a symplectic matrix with $a$, $b$, $c$, and $d$ being real parameters, and $b\neq0$.

The notion of LCT can readily be extended to $n$ dimensions, the resulting transformation being also sometimes called  $n$-dimensional FRFT. The physical interpretation is straightforward, namely a LCT is the transformation of a $n$-dimensional wave function $f(\vec{x})$ that is effected by any symplectic transformation in the $2n$-dimensional phase space of variables  $(x_1,x_2,\cdots x_n)$ and $(p_1,p_2,\cdots p_n)$. We write the symplectic matrix $\mathcal{S}$ as 
\begin{eqnarray}
\mathcal{S}&=&\begin{pmatrix}a&b\\c&d\end{pmatrix}\nonumber\\
&=&\begin{pmatrix}\mathds{1} &0\\db^{-1}&\mathds{1}  \end{pmatrix}
\begin{pmatrix}b&0\\0&b^{-1}\end{pmatrix}
\begin{pmatrix}0&\mathds{1}  \\-\mathds{1}  &0\end{pmatrix}
\begin{pmatrix}\mathds{1}  &0\\b^{-1}a&\mathds{1}  \end{pmatrix}
\label{matS}
\end{eqnarray}
where $a$, $b$, $c$, and $d$ are $n\times n$ real matrices, $\openone$ is the $n\times n$ identity matrix, and $\det(b)\neq 0$. 
Since $\mathcal{S}$ is symplectic, it obeys the constraint
\begin{eqnarray}
\mathcal{S}  J  \mathcal{S}^T = J
\mathrm{~~~~~with~}  J=\begin{pmatrix}0&\mathds{1}  \\-\mathds{1}  &0\end{pmatrix}
\end{eqnarray}
being the symplectic form, so that $\det(S)=1$.
This also implies that $ab^T$ and $cd^T$ are symmetric matrices, and $ad^T-bc^T=\mathds{1} $.
The corresponding symplectic transformation in phase space is
\begin{eqnarray}
\begin{pmatrix} \vec y \\ \vec q \end{pmatrix} 
=\mathcal{S}  \begin{pmatrix} \vec x \\ \vec p \end{pmatrix} ,
\end{eqnarray}
where $\vec y = (y_1,y_2,\cdots y_n)^T$ and $\vec q = (q_1,q_2,\cdots q_n)^T$ form a new pair of canonically conjugate $n$-tuples.
In state space, the LCT of $f(\vec x)$ can be written as
\begin{eqnarray}
\mathcal{F}_\mathcal{S}[f(\vec{x})](\vec y)&=&\frac{1}{\sqrt{(2\pi)^n|\det(b)|}}\int d\vec{x}f(\vec x)\,e^{-i\,b^{-1}\vec y\cdot \vec x  }\nonumber\\
&& ~~~~~~  \times \;
e^{\frac{i}{2}\left[  \vec x^T(b^{-1}a)\vec x+\vec y^T(db^{-1})\vec y\right]}\nonumber\\
&=&C_{db^{-1}}D_{b^{-1}}\mathcal{F}C_{b^{-1}a}[f(\vec{x})](\vec y)
\label{FRFT}
\end{eqnarray}
with
\begin{eqnarray}
C_{r}[f](\vec x)&=&e^{\frac{i}{2}\vec x^T r \vec x} f(\vec x)\nonumber\\
D_b[f](\vec x)&=&\sqrt{|\det(b)|}f(b\,\vec x)\nonumber\\
\mathcal{F}[f](\vec{x})&=&\frac{1}{(2\pi)^{n/2}}\int d\vec y f(\vec y)e^{-i\,\vec{x}\cdot\vec{y}}
\end{eqnarray}
where $C_r$ represents the chirp multiplication, $D_b$ the squeezing (or dilation) operator, and $\mathcal{F}$ the usual Fourier transform.
These operators are directly related to the decomposition of $\mathcal{S}$ in Eq.~(\ref{matS}). Note also that the chirp multiplication (in one dimension) can be expressed as a product of the other two operators, namely $C_r=R_{\theta-\pi/2}\cdot D_{\tan\theta}\cdot R_\theta$ where $R_\theta$ represent the rotation and $r=\tan\theta-\cot\theta$. Finally, note that the set of LCTs in phase space is in one-to-one correspondence with the set of Gaussian unitaries in state space \cite{weedbrook}, which can indeed be decomposed into passive linear-optics operations (phase shifters and beam splitters, i.e., rotations in phase space) and active squeezing operations (i.e., area-preserving dilations in phase space).

 Here are some properties  of LCTs that will be useful to prove our results in Section III.
 \begin{prop}
 	\leavevmode
 	\begin{enumerate}
 		\item$
 		\mathcal{F}_\mathcal{A}\mathcal{F}_\mathcal{B}=\mathcal{F}_\mathcal{AB}
$
 		\item  $
 		D_{b_1}D_{b_2}=D_{b_1b_2}
 	$
 		\item   	$
 		\mathcal{F}^{-1}=D_{-1}\mathcal{F}
$
 		where $	\mathcal{F}^{-1}$ is the inverse Fourier transform.
 		\item  $
 		\left|C_{r}\,f\right|=\left|f\right|
$
 	\end{enumerate}
 \end{prop}
 \begin{proof}
 	\leavevmode
 	\begin{enumerate}
\item  	Using Eq.~(\ref{FRFT}) and the corresponding representation in phase space, Eq.~(\ref{matS}), we see that $\mathcal{F}_\mathcal{A}\mathcal{F}_\mathcal{B}$ is represented by the matrix $\mathcal{AB}$. Since the symplectic matrices form a group, the product of two symplectic matrices is a symplectic matrix, which also admits decomposition (\ref{matS}) and thus represents the linear canonical transform $\mathcal{F_{AB}}$.
 	\end{enumerate}
 	
 	\indent \indent Proofs of {\it 2, 3} and {\it 4} are straightforward. 
 \end{proof}


 \section{Multidimensional uncertainty relations}

\subsection{Entropic uncertainty relation between two linear canonical transforms}

Let $\ket{\psi}$ be an arbitrary $n$-mode state. We wish to express the complementarity between two incompatible LCTs corresponding to two Gaussian unitaries ($U_A$ or $U_B$) applied onto $\ket{\psi}$. As shown in Fig.~\ref{transfoAB}, we measure in both cases the $n$-tuple of output $x$-quadratures, which corresponds to applying two possible $n$-modal Gaussian projective measurements on $\ket{\psi}$. The vectors of measurement outcomes are noted, respectively,  $\vec y = (y_1,\cdots, y_n)^T$ or $\vec z=(z_1,\cdots, z_n)^T$. 
Denoting as $\mathcal{A}$ and $\mathcal{B}$ the symplectic transformations associated with $U_A$ and $U_B$, and writing as
$\hat{r}=(\hat x_1,\cdots ,\hat x_n,\hat p_1,\cdots ,\hat p_n)^T$ the $2n$-dimensional vector of input quadratures,
we may express the corresponding vectors of output quadratures as
 	\begin{eqnarray}
 	\hat{r}_A =\mathcal{A} \; \hat{r} 
	\equiv \begin{pmatrix} \vec y \\ \vec q \end{pmatrix},
 	\quad
 	\hat{r}_B =\mathcal{B} \; \hat{r}
	\equiv \begin{pmatrix} \vec z \\ \vec o \end{pmatrix} .
 	\label{defyz}
 	\end{eqnarray}
where $\vec q$ (resp. $\vec o$) is the vector of quadratures that are canonically conjugate with $\vec y$ (resp. $\vec z$).
The probability distributions for $\vec y$ and $\vec z$ are thus given by the squared modulus of the LCTs associated with $U_A$ and $U_B$, 
namely $\left| \mathcal{F}_\mathcal{A}[\psi(\vec x)](\vec y) \right|^2$ and $\left| \mathcal{F}_\mathcal{B}[\psi(\vec x)](\vec z) \right|^2$.
In order to find an entropic uncertainty relation for $\vec y$ and $\vec z$, we first express the complementarity between $\mathcal{F}_\mathcal{A}$ and $\mathcal{F}_\mathcal{B}$ in the following Lemma.
 	
 

 \begin{lem}	
	\label{theoFRFT}
	Let $\mathcal{F}_\mathcal{A}$ and $\mathcal{F}_\mathcal{B}$ be two LCTs of a function $f(\vec x)$ with $\vec x = (x_1,\cdots x_n)$. Then, their squared moduli satisfy the entropic uncertainty relation
\begin{equation}
h(|\mathcal{F}_\mathcal{A}|^2)+h(|\mathcal{F}_\mathcal{B}|^2)\geq\ln\left((\pi e)^n |\det(\mathcal{B}_b\mathcal{A}_a^T-\mathcal{B}_a\mathcal{A}_b^T)|\right)
\label{EURFRFT}
\end{equation}
where 
\begin{equation}
\mathcal{A}=
\begin{pmatrix}
\mathcal{A}_a&\mathcal{A}_b\\\mathcal{A}_c&\mathcal{A}_d
\end{pmatrix}
\quad\text{and}\quad
\mathcal{B}=\begin{pmatrix}
\mathcal{B}_a&\mathcal{B}_b\\\mathcal{B}_c&\mathcal{B}_d
\end{pmatrix}
\label{matAB}
\end{equation}
are the symplectic matrices associated with $\mathcal{F}_\mathcal{A}$ and $\mathcal{F}_\mathcal{B}$ [acting on the quadrature operators as in Eq.~(\ref{defyz})] and $h(\cdot)$ denotes Shannon differential entropy.
\end{lem}

\begin{proof}
Let us define the function
\begin{equation}
G\left(\vec{x}\right)
=C_{-b^{-1}a} \, \mathcal{F}_{\mathcal{A}}(\vec x)
=e^{-\frac{i}{2}\vec{x}^{T}(b^{-1}a)\vec{x}} \, \mathcal{F}_{\mathcal{A}}\left(\vec x\right).
\end{equation}
The inverse Fourier transform of $G(\vec{x})$ is
\begin{equation}
g\left(\vec{p}\right)=\mathcal{F}^{-1}[G\left(\vec{x}\right)](\vec p)=\left[D_{-1}\mathcal{F}\,G\right]\left(\vec{p}\right),
\end{equation}
where the second equality results from property 3.
Since $\left|G\left(\vec{x}\right)\right|^2=\left|\mathcal{F}_{\mathcal{A}}\left(\vec{x}\right)\right|^2$, the probability distributions are equal, so that  $h\left(\left|G\left(\vec x \right)\right|^{2}\right)=h\left(\left|\mathcal{F}_{\mathcal{A}}\left(\vec{x}\right)\right|^{2}\right)$. Then, we may apply Eq.~(\ref{birula}) to $G$ and $g$, which gives
\begin{equation}
h\left(\left|\mathcal{F}_{\mathcal{A}}\left(\vec{x}\right)\right|^{2}\right)+h\left(\left|g\left(\vec{p}\right)\right|^{2}\right)\geq n\ln(\pi e).
\label{EURgetF}
\end{equation}
With the change of variables $\vec p\rightarrow b^{-1}\vec{p}$, we have
\begin{eqnarray}
h\left(\left|g\left(\vec{p}\right)\right|^{2}\right)&=&-\int \left|g\left(\vec{p}\right)\right|^2\ln \left|g\left(\vec{p}\right)\right|^{2}d\vec p\nonumber\\
&=& -\int \left|g\left(b^{-1}\vec{p}\right)\right|^2\ln \left|g\left(b^{-1}\vec{p}\right)\right|^{2}\frac{d\vec p}{|\det(b)|}. \qquad
\label{eq19}
\end{eqnarray}
By using the above properties of LCTs, we have
\begin{eqnarray}
\left|g\left(b^{-1}\vec{p}\right)\right|^{2} 
& = & \left|\sqrt{\det(b)}\; D_{b^{-1}}\left[g\right]\left(\vec{p}\right)\right|^{2}\nonumber\\
& = & |\det(b)| \; \big|\left[D_{b^{-1}}D_{-1}\mathcal{F}\,G\right]\left(\vec{p}\right)\big|^{2}\nonumber\\
& = & |\det(b)| \; \big|\left[D_{b^{-1}}D_{-1}\mathcal{F}\,C_{-b^{-1}a}\mathcal{F}_{\mathcal{A}}\right]\left(\vec{p}\right)\big|^{2}\nonumber\\
& = & |\det(b)| \; \big|\left[C_{-db^{-1}}D_{{-b^{-1}}}\mathcal{F}\,C_{-b^{-1}a}\mathcal{F}_{\mathcal{A}}\right]\left(\vec{p} \right)\big|^{2}\nonumber\\
& = & |\det(b)| \; \left|\left[\mathcal{F}_{\mathcal{S}_\ominus}\mathcal{F}_{\mathcal{A}}\right]\left(\vec{p}\right)\right|^{2}\nonumber\\
&=&|\det(b)|  \left|\left[\mathcal{F}_{\mathcal{S}_\ominus\mathcal{A}}\right]\left(\vec{p}\right)\right|^{2},
\end{eqnarray}
where we have defined $\mathcal{S}_\ominus=\left(\begin{array}{cc}
a & -b\\
-c & d
\end{array}\right)$, so by plugging it into Eq.~(\ref{eq19}), we get
\begin{eqnarray}
	h\left( | g\left(\vec{p}\right)|^{2}\right)
	&=&	h\left(	|\mathcal{F}_{\mathcal{S}_\ominus\mathcal{A}}\left(\vec{p}\right)|^2\right)-\ln|\det(b)|
\end{eqnarray}
since $|\mathcal{F}_{\mathcal{S}_\ominus\mathcal{A}}|^2$ is a normalized function.
Now, replacing $h\left( | g\left(\vec{p}\right)|^{2}\right)$ in Eq.~(\ref{EURgetF}), we obtain 
\begin{equation}
h\left(\left|\mathcal{F}_{\mathcal{A}}\left(\vec{x}\right)\right|^{2}\right)+h\left(	|\mathcal{F}_{\mathcal{S}_\ominus\mathcal{A}}\left(\vec{p}\right)|^2\right)\geq \ln((\pi e)^n |\det(b)|).
\label{eq22}
\end{equation}
The last step is simply to define $\mathcal{B}=\mathcal{S}_\ominus\mathcal{A}$ or equivalently $\mathcal{S}_\ominus=\mathcal{B}\mathcal{A}^{-1}$. 
Since symplectic matrices form a group and $\mathcal{A}$ and $\mathcal{B}$ are symplectic, $\mathcal{S}_\ominus$ is necessarily symplectic too.
The property that $\mathcal{A}$ is symplectic translates into
\begin{equation}
\mathcal{A}^{-1}=J\mathcal{A}^TJ^T
=\begin{pmatrix}
\mathcal{A}_d^T&-\mathcal{A}_b^T\\-\mathcal{A}_c^T&\mathcal{A}_a^T
\end{pmatrix},
\end{equation}
hence $b=\mathcal{B}_a\mathcal{A}_b^T - \mathcal{B}_b\mathcal{A}_a^T$.
Replacing $b$ into Eq.~(\ref{eq22}) completes the proof of Eq.~(\ref{EURFRFT}), which thus provides a $n$-dimensional entropic uncertainty relation for any two incompatible LCTs.
\end{proof}

Note that in the special case of one mode ($n=1$), we recover the result obtained by Guanlei et al. \cite{Guanlei2} for one-dimensional LCTs, namely
\begin{equation}
h(|\mathcal{F}_\mathcal{A}|^2)+h(|\mathcal{F}_\mathcal{B}|^2)\geq\ln\left(\pi e |a b' - a' b| \right)
\end{equation}
for two $2\times2$ matrices $\mathcal{A}=\begin{pmatrix}a&b\\c&d\end{pmatrix}$ and $\mathcal{B}=\begin{pmatrix}a'&b'\\c'&d'\end{pmatrix}$ (see Eq.~(20) in Ref.~\cite{Guanlei2}). 
Furthermore, for one-dimensional FRFTs (when $\mathcal{A}$ and $\mathcal{B}$ are simply rotations), we recover Eq.~(\ref{eq-Guanlei}) (see Eq.~(15) in Ref.~\cite{Guanlei}).
Now, back to the $n$-mode case, if we choose $\mathcal{A}=\mathds{1}$ and $\mathcal{B}$ being the direct sum of $\pi/2$ rotations on each modes (i.e., the usual $n$-dimensional Fourier transform), then $\mathcal{A}_a=\mathds{1}$, $\mathcal{A}_b=0$, $\mathcal{B}_a=0$, and $\mathcal{B}_b=\mathds{1}$, so that $\mathcal{B}_b\mathcal{A}_a^T-\mathcal{B}_a\mathcal{A}_b^T=\mathds{1}$. Hence, we get back to the original entropic uncertainty relation of Bialynicki-Birula and Mycielski, Eq.~(\ref{birula}). Finally, if we consider twice the same measurement, i.e., $\mathcal{A}=\mathcal{B}$, then
\begin{eqnarray}
\mathcal{S}_\ominus&=&\mathcal{A}\mathcal{A}^{-1}=\begin{pmatrix}
\mathcal{A}_a&\mathcal{A}_b\\\mathcal{A}_c&\mathcal{A}_d
\end{pmatrix} \begin{pmatrix}
\mathcal{A}_d^T&-\mathcal{A}_b^T\\-\mathcal{A}_c^T&\mathcal{A}_a^T
\end{pmatrix}\nonumber\\
&=&\begin{pmatrix}
\mathcal{A}_a\mathcal{A}_d^T-\mathcal{A}_b\mathcal{A}_c^T&\,\,
-\mathcal{A}_a\mathcal{A}_b^T+\mathcal{A}_b\mathcal{A}_a^T\\
\mathcal{A}_c\mathcal{A}_d^T-\mathcal{A}_d\mathcal{A}_c^T&\,\,
-\mathcal{A}_c\mathcal{A}_b^T+\mathcal{A}_d\mathcal{A}_a^T
\end{pmatrix}.\quad
\end{eqnarray}
But since  $\mathcal{S}_\ominus=\mathds{1}$, we have $\mathcal{A}_b\mathcal{A}_a^T-\mathcal{A}_a\mathcal{A}_b^T=~0$, so that the lower bound in Eq.~(\ref{EURFRFT}) is $-\infty$. This means that we have no lower limit on the entropy $h(|\mathcal{F}_\mathcal{A}|^2)$ so the probability distribution $|\mathcal{F}_\mathcal{A}|^2$ can be arbitrarily narrow, as expected.

Interestingly, in the special case where $\mathcal{A}=\openone$, it is possible to find a simpler alternative proof of Lemma~\ref{theoFRFT}.
We define 
\begin{equation}
\mathcal{S}=\begin{pmatrix}
\mathcal{B}_a&\mathcal{B}_b\\-(\mathcal{B}_b^{-1})^T&0
\end{pmatrix}.
\label{defS}
\end{equation}
and may easily check that $\mathcal{S}$ is a symplectic matrix by verifying that $\mathcal{S}J\mathcal{S}^T=J$. Indeed
\begin{equation}
\mathcal{S}J\mathcal{S}^T=\begin{pmatrix}
-\mathcal{B}_b\mathcal{B}_a^T+\mathcal{B}_a\mathcal{B}_b^T&\mathcal{B}_b\mathcal{B}_b^{-1}\\-(\mathcal{B}_b^{-1})^T\mathcal{B}_b^T&0 
\end{pmatrix} ,
\end{equation}
is equal to $J$ since $\mathcal{B}_a\mathcal{B}_b^T$ is a symmetric matrix and $\det(\mathcal{B}_b)\ne 0$ (as $\mathcal{B}$ is also a symplectic matrix).
Thus, $\mathcal{S}$ transforms $\hat r$ into a new vector of quadratures,
\begin{equation}
\mathcal{S}\hat r=\begin{pmatrix}
\vec{z}\\-(\mathcal{B}_b^{-1})^T \vec{x}
\end{pmatrix}
\end{equation}
where $\vec z$ is the vector of position quadratures in $\hat r_B$ [see Eq.~(\ref{defyz})].
Since $\mathcal{S}$ is symplectic, $\vec{z}$ and $-(\mathcal{B}_b^{-1})^T \vec{x}$ are two vectors of canonically conjugate quadratures, which we may plug into Eq.~(\ref{birula}), giving 
\begin{equation}
h\left(-(\mathcal{B}_b^{-1})^T \vec{x}\right) + h(\vec{z})  \geq n\ln (\pi e).
\label{eqbx}
\end{equation}
By using the scaling property of the differential entropy, we have $h(-(\mathcal{B}_b^{-1})^T \vec{x})=h(\vec{x})+\ln(|\det ((\mathcal{B}_b^{-1})^T )|)$, so that Eq.~(\ref{eqbx}) becomes
\begin{equation}
h(\vec{x})+h(\vec{z})\geq \ln \left((\pi e)^n|\det(\mathcal{B}_b)|\right).
\label{notgeneralEUR}
\end{equation}
Since the probability distribution of $\vec x$ is  
 $\left| \mathcal{F}_\mathcal{\openone}\right|^2$ and that of $\vec z$ is $\left| \mathcal{F}_\mathcal{B}\right|^2$,
we recover Lemma~\ref{theoFRFT} when $\mathcal{A}=\mathds{1}$.


\subsection{Entropic uncertainty relation based on a matrix of commutators}

Lemma~\ref{theoFRFT} provides an entropic uncertainty relation for any two $n$-dimensional LCTs, $\mathcal{F}_\mathcal{A}$ and $\mathcal{F}_\mathcal{B}$. As we show in the following theorem, this uncertainty relation can also be expressed in terms of a matrix of commutators between the measured variables.  This is our main result.

\begin{theo}	
	\label{theoFRFTcom}
	Let $\vec y = (\hat y_1,\cdots \hat y_n)^T$ be a vector of commuting quadratures and $\vec z = (\hat z_1,\cdots \hat z_n)^T$ be another vector of commuting quadratures. Let the components of $\vec y$ and $\vec z$ be written each as a linear combination of the $(\hat x, \hat p)$ quadratures of a $n$-modal system, namely
\begin{eqnarray}
\hat y_i &=& \sum_{k=1}^{n} a_{i,k} \, \hat x_k  + \sum_{k=1}^{n} a'_{i,k} \, \hat p_k  \qquad (i=1,\cdots n) \nonumber \\
\hat z_j &=& \sum_{k=1}^{n} b_{j,k} \, \hat x_k  + \sum_{k=1}^{n} b'_{j,k} \, \hat p_k   \qquad (j=1,\cdots n).
\end{eqnarray}
Then, the probability distributions of the vectors of jointly measured quadratures $\hat y_i$'s or $\hat z_j$'s satisfy the entropic uncertainty relation 
	\begin{equation}
	h(\vec y)+h(\vec z)\geq\ln\left((\pi e)^n |\det K|\right)
	\label{EURFRFTcom}
	\end{equation}
	where  $K_{ij}=[\hat y_i,\hat z_j]$ denotes the $n\times n$ matrix of commutators (which are scalars) and $h(\cdot)$ denotes Shannon differential entropy.
\end{theo}
\begin{proof}
Since the quadratures $\hat y_i$ commute, $[\hat y_i,\hat y_j]=0$, $\forall i,j$,  they can be jointly measured, and similarly for the $\hat z_j$'s. Thus, the $n$ measured quadratures correspond here to the output of $\mathcal{F}_\mathcal{A}$ or $\mathcal{F}_\mathcal{B}$ described by the symplectic matrix $\mathcal{A}$ or $\mathcal{B}$, as defined in Eq.~(\ref{matAB}). We simply have to compute the commutator between quadrature $\hat y_i$ (at the output of $\mathcal{F}_\mathcal{A}$) and $\hat z_j$ (at the output of $\mathcal{F}_\mathcal{B}$):
\begin{eqnarray}
K_{ji}&=&[\hat y_j,\hat z_i]\nonumber\\
&=&\sum_{k=1}^{2n}\sum_{m=1}^{2n}\mathcal{A}_{jk}\mathcal{B}_{im}[\hat r_k,\hat r_m]\nonumber\\
&=&i\sum_{m=1}^{2n}\left(\sum_{k=1}^{n}\mathcal{A}_{jk}\mathcal{B}_{im}\delta_{m,k+n}-\sum_{k=n+1}^{2n}\mathcal{A}_{jk}\mathcal{B}_{im}\delta_{m,k-n}\right)\nonumber\\
&=&i\left(\sum_{k=1}^{n}\mathcal{A}_{jk}\mathcal{B}_{i,k+n}-\sum_{k=n+1}^{2n}\mathcal{A}_{jk}\mathcal{B}_{i,k-n}\right)\nonumber\\
&=&i\left(\sum_{k=1}^{n}\mathcal{A}_{jk}\mathcal{B}_{i,k+n}-\mathcal{A}_{j,k+n}\mathcal{B}_{ik}\right)\nonumber\\
&=&i\left(\sum_{k=1}^{n}(\mathcal{A}_a)_{jk}(\mathcal{B}_b)_{ik}-(\mathcal{A}_b)_{jk}(\mathcal{B}_a)_{ik}\right)\nonumber\\
&=& i\left(\mathcal{B}_b\mathcal{A}_a^T-\mathcal{B}_a\mathcal{A}_b^T\right)_{ij}
\end{eqnarray}
Using Lemma~\ref{theoFRFT}, we know that the probability distributions $|\mathcal{F}_\mathcal{A}(\vec y)|^2$ and $|\mathcal{F}_\mathcal{B}(\vec z)|^2$  satisfy the entropic uncertainty relation, Eq.~(\ref{EURFRFT}). Since $\mathcal{B}_b\mathcal{A}_a^T-\mathcal{B}_a\mathcal{A}_b^T=-iK^T$, we conclude that $|\det(\mathcal{B}_b\mathcal{A}_a^T-~\mathcal{B}_a\mathcal{A}_b^T)|=|\det(K)|$, which concludes the proof of Eq.~(\ref{EURFRFTcom}).
\end{proof}

We now show that this result holds even if we jointly measure $n$ quadratures on a larger-dimensional system.

\begin{theobis}	{\rm (Extended version.)}
	Let $\vec y_n = ~(\hat y_1,\cdots \hat y_n)^T$ be a vector of commuting quadratures and $\vec z_n = ~(\hat z_1,\cdots \hat z_n)^T$ be another vector of commuting quadratures. Let each components of the
	$\vec y_n$ and $\vec z_n$ be written  as a linear combination of the $(\hat x, \hat p)$ quadratures of a $N$-modal system with $N>n$, namely
\begin{eqnarray}
\hat y_i &=& \sum_{k=1}^{N} a_{i,k} \, \hat x_k  + \sum_{k=1}^{N} a'_{i,k} \, \hat p_k  \qquad (i=1,\cdots n) \nonumber \\
\hat z_j &=& \sum_{k=1}^{N} b_{j,k} \, \hat x_k  + \sum_{k=1}^{N} b'_{j,k} \, \hat p_k   \qquad (j=1,\cdots n)
\end{eqnarray}
Then, the probability distributions of the vectors of jointly measured quadratures $\hat y_i$'s or $\hat z_j$'s satisfy the entropic uncertainty relation 
	\begin{equation}
	h(\vec y_n)+h(\vec z_n)\geq\ln\left((\pi e)^n |\det K|\right)
	\label{EURFRFTcombis}
	\end{equation}
	where  $K_{ij}=[\hat y_i,\hat z_j]$ denotes the $n\times n$ matrix of commutators (which are scalars) and $h(\cdot)$ denotes Shannon differential entropy.
\end{theobis}

\begin{proof}
The $N$-dimensional vectors $\vec y$ and  $\vec z$ can be decomposed as 
\begin{eqnarray}
\vec y = \begin{pmatrix}
\vec y_n \\ \vec y_{>}
\end{pmatrix},
\qquad
\vec z = \begin{pmatrix}
\vec z_n \\ \vec z_{>}
\end{pmatrix},
\end{eqnarray}
with $\vec y_n=(y_1,\cdots y_n)^T$ or $\vec z_n=(z_1,\cdots z_n)^T$ being the $n$ measured quadratures, while 
$\vec y_{>}=(y_{n+1},\cdots y_N)^T$ or $\vec z_{>}=(z_{n+1},\cdots z_N)^T$ are being traced over.
We write
\begin{eqnarray}
\hat y_i = \sum_{k=1}^{2N} \mathcal{A}_{i,k} \, \hat r_k  \qquad 
\hat z_j = \sum_{k=1}^{2N} \mathcal{B}_{j,k} \, \hat r_k  
\end{eqnarray}
with $i,j=1,\cdots n$, which generalizes Eq.~(\ref{defyz}) in the case where $\hat r$ is a $2N$-dimensional vector and 
$\mathcal{A}$ and $\mathcal{B}$ are $2N\times 2N$ symplectic matrices (we only need to specify  the upper block of size $n\times 2N$ of $\mathcal{A}$ and $\mathcal{B}$, which defines $\hat y_i$ or $\hat z_j$, and complete the matrices by ensuring that they remain symplectic).

We first note that the right-hand side term of Eq.~(\ref{EURFRFTcombis}) is invariant under symplectic transformations (if both symplectic matrices $\mathcal{A}$ or $\mathcal{B}$ are multiplied by a same symplectic matrix). Indeed, the commutation relations are preserved along symplectic transformations and the determinant is invariant under permutations (the order of the quadratures is irrelevant), hence $\det(K)$ is invariant. Thus, we may always apply some symplectic transformation on the $N$ modes so that the measured quadratures in the first case are $y_i=x_i$, with $i=1,\cdots n$. The two upper blocks of matrix $\mathcal{A}$ are then given by
\begin{equation}
\mathcal{A}_a=\begin{pmatrix}
\mathds{1}_{n\times n}&0_{n\times (N-n)}\\
\cdots &\cdots
\end{pmatrix}_{N\times N}
\end{equation}
and 
\begin{equation}
\mathcal{A}_b=\begin{pmatrix}
0_{n\times n}&0_{n\times (N-n)}\\
\cdots &\cdots
\end{pmatrix}_{N\times N}
\end{equation}
where we do not need to specify the matrix elements denoted with a dot.


Next, we may assume with no loss of generality that the two upper blocks of $\mathcal{B}$ are given by
\begin{equation}
\mathcal{B}_a=\begin{pmatrix}
B_{n \times n}&C_{n\times (N-n)}\\D_{(N-n) \times n}&\mathds{1}_{(N-n)\times (N-n)}
\end{pmatrix}_{N\times N}
\end{equation}
and
\begin{equation}
\mathcal{B}_b=\begin{pmatrix}
B'_{n \times n}&C'_{n\times (N-n)}\\0_{(N-n) \times n}&\mathds{1}_{(N-n)\times (N-n)}
\end{pmatrix}_{N\times N} ,
\end{equation}
with $B$ and $C$ containing all $b_{j,k}$ entries for $j=1,\cdots,n$ and $k=1,\cdots,N$, and $ B'$ and $C'$ containing all $b'_{j,k}$ entries for $j=1,\cdots,n$ and $k=1,\cdots,N$. 
This is the case because the last $N-n$ quadratures $\vec z_>$ are traced over, so they may be chosen arbitrarily as long as $\mathcal{B}$ remains symplectic.
This means that we must check that $\mathcal{B}_a\mathcal{B}_b^T$ is symmetric, which implies that
\begin{equation}
D=\left(C-C'\right)^T (B')^{-T}
\end{equation}
where $(\cdot)^{-T}$ stands for the inverse of the transpose of the matrix. 
Thus, the matrix $D$ can always be chosen in order to ensure that $\mathcal{B}$ is symplectic.
 It is easy to write the $n\times n$ restricted matrix of commutators of the measured quadratures $\vec y_n$ and $\vec z_n$ (the first $n$ quadratures of $\vec y$ and $\vec z$), giving $|\det K| = |\det B'|$, so the inverse of $B'$ is well defined as long as $\det K \ne 0$.


At this point, we only need to prove Eq.~(\ref{EURFRFTcombis}) in case the upper blocks of $\mathcal{A}$ and $\mathcal{B}$ are defined as above and 
 $|\det K|$ is replaced by $|\det B'|$. As before, we define the symplectic matrix
\begin{equation}
\mathcal{S}=\begin{pmatrix}
\mathcal{B}_a&\mathcal{B}_b\\-(\mathcal{B}_b)^{-T}&0
\end{pmatrix}.
\end{equation}
It transforms the vector of quadratures $\hat r$ into 
\begin{equation}
\mathcal{S}\, \hat r
=\begin{pmatrix}
\mathcal{B}_a \, \vec x + \mathcal{B}_b \, \vec p   \\  -(\mathcal{B}_b)^{-T} \vec x
\end{pmatrix}
=\begin{pmatrix}
\vec z_n   \\ \vec z_>  \\-(B')^{-T} \vec x_n \\ (C')^T (B')^{-T} \vec x_n - \vec x_>
\end{pmatrix}
\end{equation}
where $\vec x_n=(x_1,\cdots x_n)^T$, $\vec x_{>}=(x_{n+1},\cdots x_N)^T$.
This implies that  $\vec z_n$ and $-(B')^{-T} \vec x_n$ are canonically conjugate $n$-tuples, so that we may apply Eq.~(\ref{birula}) on the reduced state of the first $n$ modes, giving
\begin{equation}
h\left(-(B')^{-T} \vec{x_n}\right) + h(\vec{z_n})  \geq n\ln (\pi e).
\end{equation}
Using the scaling property of the differential entropy $h(-(B')^{-T} \vec{x_n})=h(\vec{x_n})+\ln(|\det (B')^{-T} |)$, we obtain
\begin{equation}
h(\vec x_n)+h(\vec z_n)\geq\ln((\pi e)^n |\det B'|)
\label{conjecture1}
\end{equation}
This implies Eq.~(\ref{EURFRFTcombis}), thus completing the proof of the extended version of Theorem~\ref{theoFRFTcom}.
\end{proof}

%

Interestingly, Eq.~(\ref{EURFRFTcombis}) coincides with Eq.~(\ref{huang}) in the special case $n=1$. Thus, our Theorem~\ref{theoFRFTcom}
can be viewed as an extension of the result by Huang \cite{huang} when we measure more than one mode ($n>1$). As already mentioned, we can check that if $\mathcal{A}=\mathds{1}$ and $\mathcal{B}$ is a direct sum of $\pi/2$-rotations on each modes (i.e., the usual Fourier transform), then $K=-i\mathds{1}$ and we recover Bialynicki-Birula and Mycielski relation, Eq.~(\ref{birula}).

\subsection{Extension to Rényi entropies}

The Shannon differential entropy is a special case of the family of Rényi differential entropies defined as
\begin{equation}
	h_\alpha(|f(\vec x)|^2)=\frac{1}{1-\alpha}\ln\left(\int d\vec x\,(|f(\vec x)|^{2})^\alpha\right)
\label{eq-def-Renyi}	
\end{equation}
when $\alpha\rightarrow1$. 
Let us now derive generalized entropic uncertainty relations for these entropies.
\begin{theo}
\label{theo-Renyi}	
Let $\vec y = (\hat y_1,\cdots \hat y_n)^T$ be a vector of commuting quadratures, $\vec z = (\hat z_1,\cdots \hat z_n)^T$ be another vector of commuting quadratures, and let the components of these vectors be written each as a linear combination of the $(\hat x, \hat p)$ quadratures of a $N$-modal system ($N \ge n$).
Then, the probability distributions of the vectors of jointly measured quadratures $\hat y_i$'s or $\hat z_j$'s satisfy the Rényi entropic uncertainty relation 
	\begin{eqnarray}
	h_\alpha(\vec y)+h_\beta(\vec z)&\geq&\frac{n\ln(\alpha)}{2\left(\alpha-1\right)}+\frac{n\ln(\beta)}{2\left(\beta-1\right)}\nonumber\\
	&& ~~ +n\ln(\pi)+ \ln\left|\det K \right|.
	\label{renyiEURFRFT}
	\end{eqnarray}
	where 
	\begin{equation}
\frac{1}{\alpha}+\frac{1}{\beta}=2,  \qquad \alpha>0, \qquad \beta>0,
\label{alphabeta}
	\end{equation}
	$K_{ij}=[\hat y_i,\hat z_j]$ is the matrix of commutators (which are scalars), and $h_\alpha(\cdot)$ is the Rényi differential entropy as defined in Eq.~(\ref{eq-def-Renyi}).
\end{theo}
\begin{proof}
The proof follows exactly the same steps as the proof of Lemma \ref{theoFRFT} and Theorem \ref{theoFRFTcom} (both versions) except that Eq.~(\ref{birula}) is replaced by its counterpart for Rényi entropies \cite{birula2}
\begin{equation}
h_\alpha(\vec x)+h_\beta(\vec p)\geq n\ln(\pi)+\frac{n\ln(\alpha)}{2\left(\alpha-1\right)}+\frac{n\ln(\beta)}{2\left(\beta-1\right)}
\end{equation}
for $(\alpha, \beta)$ satisfying Eq.~(\ref{alphabeta}). Note that Rényi entropies also verify the scaling property $h_\alpha(S\vec{x})=h_\alpha(\vec{x})+\ln|S|$.
\end{proof}
	As expected, in the limit where $\alpha\rightarrow1$ and $\beta\rightarrow1$, we recover our uncertainty relations for Shannon differential entropies. Also, in a one-dimensional case ($N=n=1$), Eq. (\ref{renyiEURFRFT}) coincides with the result found in \cite{Guanlei2}.
	
\subsection{Covariance-based uncertainty relation}

Finally, by exploiting Theorem \ref{theoFRFTcom}, it is also possible to derive an uncertainty relation in terms of covariance matrices. This can been viewed as a $n$-dimensional extension of the usual Robertson uncertainty relation in position and momentum spaces where, instead of expressing the complementarity between observables $\hat A$ and $\hat B$ (which are linear combinations of quadratures), namely
\begin{equation}
\Delta \hat A \, \Delta \hat B  \ge | [\hat A,\hat B] |  /2
\end{equation}
with $ [\hat A,\hat B]$ being a scalar, we consider the complementarity between two $n$-tuples of commuting observables.

\begin{theo}
\label{theo-covariance}
Let $\vec y = (\hat y_1,\cdots \hat y_n)^T$ be a vector of commuting quadratures, $\vec z = (\hat z_1,\cdots \hat z_n)^T$ be another vector of commuting quadratures, and let the components of these vectors be written each as a linear combination of the $(\hat x, \hat p)$ quadratures of a $N$-modal system ($N \ge n$).
Let $\gamma^\mathcal{A}_{ij}=\langle\{\hat y_i,\hat y_j\}\rangle/2-\langle\hat y_i\rangle\langle\hat y_j\rangle$ and $\gamma^\mathcal{B}_{ij} =~\langle\{\hat z_i,\hat z_j\}\rangle/2-\langle\hat z_i\rangle\langle\hat z_j\rangle$ be the (reduced) covariance matrices of the $\hat y_i$ and $\hat z_i$ quadratures. Then
	\begin{eqnarray}
\left( \det \gamma^\mathcal{A} \right)^{1\over 2} \, \left( \det \gamma^\mathcal{B}\right)^{1\over 2}   \geq\frac{|\det K|}{2^n}
	\label{URcovariance}
	\end{eqnarray}
	where  $K_{ij}=[\hat y_i,\hat z_j]$ denotes the commutator matrix.
\end{theo}
\begin{proof}
Let us define the entropy powers of $\vec y$ and $\vec z$ as
\begin{equation}
N_\mathcal{A}=\frac{1}{2\pi e}e^{\frac{2}{n}h(\vec y)},\quad
N_\mathcal{B}=\frac{1}{2\pi e}e^{\frac{2}{n}h(\vec z)},
\end{equation}
which allows us to rewrite Eq.~(\ref{EURFRFTcom}) as an entropy-power uncertainty relation (see \cite{hertz})
\begin{equation}
N_\mathcal{A}N_\mathcal{B}\geq \frac{ |\det K|^{2/n} }{4} .
\label{entropy-power-ur}
\end{equation} 
Since the maximum entropy for a fixed covariance matrix is reached by the Gaussian distribution, we have that $N_\mathcal{A}\leq~(\det \gamma^\mathcal{A})^{1/n}$ and $N_\mathcal{B}\leq~(\det \gamma^\mathcal{B})^{1/n}$. Combining these inequalities with Eq.~(\ref{entropy-power-ur}), we prove our theorem.
\end{proof}

In the one-mode case, we obtain $\Delta \hat y_1 \Delta \hat z_1 \geq |[\hat y_1,\hat z_1]| /2$ which is Robertson uncertainty relation applied to the two quadratures $\hat y_1$ and $\hat z_1$, as already mentioned. Thus, Theorem \ref{theo-covariance} extends this relation to two joint measurements of $n$ modes and accounts for the correlations between the $y_i$'s  via the term $\det \gamma^\mathcal{A}$  (as well as between the $z_j$'s via the term $\det \gamma^\mathcal{B}$). Note, however, that this covariance-based uncertainty relation is less strong than the entropic uncertainty relation since Theorem~\ref{theo-covariance}
 follows from Theorem~\ref{theoFRFTcom}. 

\section{Conclusions}

We have derived an entropic uncertainty relation which applies to any two $n$-dimensional LCTs $\mathcal{F}_\mathcal{A}(\vec y)$ and $\mathcal{F}_\mathcal{B}(\vec z)$ or any two $n$-modal Gaussian projective measurements resulting in outcomes $\vec y$ and $\vec z$. As implied by our Theorem \ref{theoFRFTcom}, the sum of the entropy of the probability distributions for $\vec y$ and $\vec z$ is lower bounded by a quantity that depends on the determinant of the matrix of commutators $[\hat y_i,\hat z_j]$, a quantity that is invariant under symplectic transformations. This is a generalization of the usual entropic uncertainty relation (\ref{birula}) due to Bialynicki-Birula and Mycielski in the case of any two $n$-dimensional observables that are not canonically conjugate but are connected by an arbitrary LCT.

Theorem \ref{theoFRFTcom} can also be viewed as a natural extension of the uncertainty relation~(\ref{huang}) due to Huang \cite{huang}.
As shown in Figure \ref{transfoAB}, the two considered measurements can be realized by applying a Gaussian unitary $U_A$ or $U_B$ before measuring the $\hat x$ quadratures.  If we restrict ourselves to measuring the $\hat x$ quadrature of the first mode only, then the resulting quadrature is $\hat{A}$ or $\hat{B}$ as defined in Eq.~(\ref{ABdef}). Thus, our entropic uncertainty relation generalizes Huang's setup by including the measurement of any number of modes  instead of the first one only. It naturally accounts for the correlations between the measured $y_i$'s (as well as $z_j$'s) via the use of joint entropies.
Following the same scheme, we also recover the usual entropic uncertainty relation~(\ref{birula}) by applying either the identity ($U_\mathcal{A}=\mathds{1}$) or a tensor product of $\pi/2$ rotations on all mode ($U_\mathcal{B}=R_{\pi/2}^{\otimes n}$) before measuring all $x$ quadratures.


Our results still hold true (with some adaptations) when Shannon entropies are replaced by Rényi entropies, as proven in Theorem~\ref{theo-Renyi}. They also  imply a generalized version of Robertson uncertainty relation expressing the complementarity between two $n$-tuples of quadrature observables in terms of the determinant of a commutator matrix, see Theorem~\ref{theo-covariance}.


As a final note, it must be stressed that we have restricted ourselves to observables that are linear combinations of the $\hat x$ and $\hat p$ quadratures throughout this work, which implies that all commutators $[\hat y_i,\hat z_j]$ are scalars, as well as $\det K$.  However, we believe that it should be possible to extend Theorem \ref{theoFRFTcom} to general vectors of commuting Hermitian operators $\vec A$ and $\vec B$. Then, all commutators would be replaced by their mean values, in analogy with the usual Robertson relation. We therefore suggest the following conjecture:

\begin{conj}
Let $\vec A=(A_1,\cdots A_n)$ be a vector of commuting observables,  $\vec B=(B_1,\cdots B_n)$ be  another vector of commuting observables and $\ket{\psi}$ be the state of the system. The probability distributions of the jointly measured observables $A_i$'s or $B_j$'s in state $\ket{\psi}$ satisfy the entropic uncertainty relation
\begin{equation}
h(\vec A)+h(\vec B)\geq\ln\left((\pi e)^n |\det \bra{\psi}K\ket{\psi} |\right)
\end{equation}
where $K_{ij}=[A_i,B_j]$
\end{conj}
This would be a further generalization of the entropic uncertainty relation, also implying an extended Robertson relation involving a matrix of mean values of commutators instead of Eq.~(\ref{URcovariance}). Investigating this conjecture is an interesting topic of future work.

\medskip
\noindent {\it Acknowlegments}: 
We thank Michael G. Jabbour for useful discussions. This work was supported by the F.R.S.-FNRS Foundation under Project No. T.0199.13.
A.H. acknowledges financial support from the F.R.S.-FNRS Foundation.



\begin{thebibliography}{99}
	

	
	
	\bibitem{Heisenberg} W. Heisenberg, Z. Phys. {\bf 43}, 172 (1927).

	\bibitem{Kennard} E. H. Kennard, Z. Physik {\bf44}, 326 (1927).


	\bibitem{Schrodinger} E. Schr\"{o}dinger,	Preuss. Akad. Wiss. {\bf 14}, 296 (1930).
	\bibitem{Robertson} H. P. Robertson, Phys. Rev. {\bf 35}, 667A (1930).
	

	\bibitem{Hirschman} I. I. Hirschman, Am. J. Math. {\bf 79}, 152 (1957).

	\bibitem{Birula} I. Bialynicki-Birula and J. Mycielski, Commun. Math. Phys. {\bf 44}, 129 (1975).
	\bibitem{Beckner} W. Beckner, Ann. Math. {\bf 102}, 159 (1975).

	 \bibitem{Coles} P. J. Coles, M. Berta, M. Tomamichel, and S. Wehner, Rev. Mod. Phys. {\bf 89}, 15002 (2017).
	
	
	\bibitem{hertz}  A. Hertz, M. G. Jabbour, and N. J. Cerf, J. Phys. A {\bf 50}, 385301 (2017).
	
	\bibitem{huang}  Y. Huang, Phys. Rev. A {\bf 83} 052124 (2011).
	
	\bibitem{Guanlei} X. Guanlei, W. Xiaotong and X. Xiaotong, Signal Process. {\bf 89} 2692 (2009).
	
	\bibitem{Bultheel} A. Bultheel and H. Martinez-Sulbaran, Bull. Belg. Math. Soc. {\bf 13} 971 (2006).
	
	\bibitem{Pei} S.C. Pei and J.J Ding, IEEE Trans. Sig. Proc. {\bf 49} (4), 878 (2001).
	
	\bibitem{Morsche} H.G. ter Morsche and P.J. Oonincx, Tecnical Report PNA-R9919, CWI, Amsterdam (1999).
	
	\bibitem{Ozaktasbook} H.M. Ozaktas, Z. Zalevsky and M.A. Kutay. {\it The fractional Fourier tansform}. Wiley, Chichester (2001).
	
	\bibitem{Ozaktas}H.M. Ozaktas and M.A. Kutay, Adv. Imag. Elect. Phys. {\bf 106} 239 (1999).)
	
	
	\bibitem{weedbrook} C. Weedbrook, S. Pirandola, R. Garcia-Patron, N. J. Cerf, T. C. Ralph, J. H. Shapiro, and S. Lloyd, Rev. Mod. Phys. {\bf 84}, 621 (2012).
	
	
	\bibitem{Guanlei2} X. Guanlei, W. Xiaotong and X. Xiaotong, IET Signal Process. {\bf 3} (5) 392– 402 (2009).

\bibitem{birula2}I. Bialynicki-Birula, Phys. Rev. A {\bf 74} 052101 (2006).





%
%
%
%
%
%
%
%
%
%
	
\end{thebibliography}
\end{document}